\documentclass[11pt]{article}
\usepackage[T1]{fontenc}
\usepackage{lmodern}
\usepackage[margin=1in]{geometry}
\usepackage{amsmath,amssymb,amsthm,mathtools}
\usepackage{booktabs,tabularx,array}
\usepackage{enumitem,setspace}
\usepackage{needspace}
\usepackage{microtype}
\usepackage[authoryear,round]{natbib}
\usepackage[hidelinks]{hyperref}
\newtheorem{theorem}{Theorem}[section]
\newtheorem{proposition}[theorem]{Proposition}
\newtheorem{lemma}[theorem]{Lemma}
\newtheorem{corollary}[theorem]{Corollary}
\theoremstyle{definition}

\newtheorem{condition}{Condition}

\theoremstyle{remark}

\newcommand{\E}{\mathbb E}
\newcommand{\Prb}{\mathbb P}

\newcommand{\cP}{\mathcal P}
\newcommand{\cT}{\mathcal T}
\newcommand{\cB}{\mathcal B}
\newcommand{\R}{\mathsf R}
\newcommand{\Sact}{\mathsf S}
\newcommand{\M}{\mathsf M}
\newcommand{\Hact}{\mathsf H}
\hypersetup{pdftitle={Reputation without a Control Group},pdfauthor={Shubh Lashkery and Georgy Lukyanov},pdfsubject={Reputation, delegated implementation, and intergenerational learning},pdfkeywords={reputation, career concerns, organizational learning, knowledge transmission}}
\title{Reputation without a Control Group}
\author{Shubh Lashkery \and Georgy Lukyanov}
\date{September 2026}
\begin{document}
\maketitle
\begin{abstract}
An adviser who warns that a task is difficult may become harder to evaluate when her advice is followed more thoroughly. We study a long-lived adviser and successive short-lived workers who choose between standard and intensive implementation. Standard implementation reveals whether the warning was correct and gives the next worker an opportunity to acquire cost-saving practical knowledge. Intensive implementation protects the project but produces only occasional evidence about the adviser. We construct a stationary sequential equilibrium in which the adviser initially accepts an informative implementation, withholds the next project after her reputation improves, and resumes recommendations once inherited know-how has been lost. The interruption is chosen because it changes the successor's implementation decision: preserving know-how reverses the adviser's preference at the relevant history. All realized evidence remains public. Weak subsequent evidence eventually ends the low-ability adviser's protection, but this can take a long time.
\end{abstract}
\medskip
\noindent\textit{Keywords:} reputation; career concerns; delegated implementation; organizational learning; knowledge transmission.

\section{Introduction}\label{sec:intro}

A technical adviser who assigns a project to a junior engineer may warn that the task will require unusually intensive work. The engineer faces a familiar choice. She can try to carry out the task using the standard procedure, or she can incur the additional expense of a more intensive implementation. If she trusts the adviser's assessment, she is more likely to choose the latter. This response may improve the project's chances of success. At the same time, it makes it harder to tell whether the warning was justified: a successful project does not by itself reveal whether the extra resources were necessary.

To make the choice more concrete, think of a team that can adapt a familiar design in-house or purchase a more expensive, fully supported solution. The second option protects the project against the difficulty identified by the adviser. But it neither shows whether the ordinary design would have worked nor gives the incoming engineer the same experience of adapting that design herself. We use this as a stylized example of two implementation technologies; intensive implementation need not mean doing more of exactly the same work.

The present paper asks how this problem changes when the people carrying out the advice learn from their predecessors. Consider a firm in which successive project teams work with the same technical adviser. A junior engineer can acquire practical knowledge by taking part in the preceding team's work. Having this knowledge makes it less costly to use the standard procedure on her own assignment. By contrast, if that procedure is not used while the two teams overlap, the incoming engineer must reconstruct it at an additional cost. A written account of the earlier outcome may remain available, but it does not provide the same practical preparation.

These two observations create a tension. On the one hand, the adviser benefits when the current team has useful knowledge: it can carry out the project at a lower cost. On the other hand, the same knowledge makes the team more willing to try the standard procedure, whose outcome provides a relatively informative assessment of her warning. Moreover, using that procedure gives the next team an opportunity to learn it. The adviser therefore has a reason to consider not only the outcome of the current project, but also the kind of implementation that will remain attractive to future teams.

We show that this consideration can make a temporary interruption of advice profitable. The adviser may first put forward a project and allow her warning to be tested. After an outcome that improves her reputation, she withholds the next project. The interruption prevents the transfer of practical knowledge to the incoming worker. When she subsequently resumes recommending projects, that worker chooses intensive implementation. The warning is the same, and the public belief about the adviser has not changed during the interruption. What has changed is the worker's cost of using the alternative procedure.

It is worth explaining why the interruption matters. If the incoming worker retained the same practical knowledge, simply postponing the project would leave the adviser facing the same informative implementation tomorrow. In the equilibrium we construct, this postponement would not be worthwhile: the adviser would prefer to put the project forward today. She withholds it because the interruption changes the next worker's incentives. The result is thus not exhausted by the observation that somebody who cares about her reputation may avoid an unfavorable test.

Our model has one long-lived adviser and a sequence of short-lived workers. The adviser's fixed diagnostic ability is unknown both to her and to the public. At each advisory opportunity she can put forward a project accompanied by a warning that standard implementation may be insufficient, or withhold the project. If it is put forward, the worker chooses between standard and intensive implementation. Standard implementation fails precisely when the warning is correct. Intensive implementation succeeds under either realization of project difficulty, although a separate, occasional assessment still provides some information about the adviser. The distinction is between two productive choices, not between undertaking a project and conducting a preliminary trial.

The workers differ in whether they have inherited practical knowledge of the standard procedure. A worker without this knowledge can still use it, but must pay an additional reconstruction cost. During standard implementation, the incoming worker can pay a privately observed learning cost and acquire the knowledge for her own project. She does so for her own benefit; the model does not require an outgoing worker to sacrifice resources for a successor she will never meet again. Nor do we ask the adviser to subsidize training.

For intermediate reputations, the reconstruction cost changes the implementation decision. An experienced worker uses the standard procedure, whereas an inexperienced worker chooses intensive implementation. This is the region in which the adviser may gain from interrupting the transfer of knowledge. At lower reputations, even an inexperienced worker finds the standard procedure attractive. At higher reputations, both kinds of worker choose intensive implementation. Thus practical knowledge affects reputational exposure only over a particular range of beliefs.

The equilibrium result has three parts. First, we give sufficient conditions under which standard implementation renews practical knowledge, and construct a stationary equilibrium with the interruption described above. Second, we isolate the value of losing that knowledge by comparing actual withholding with a one-period postponement that preserves it. The two comparisons have opposite signs at the relevant history. We also change the transition rule permanently and solve the adviser's problem again; preserving knowledge after every postponement removes the same interruption. Third, we show that the subsequent period of intensive implementation is not permanent. The occasional assessments remain informative, and under low ability they eventually lead workers to resume standard implementation.

The existence proof uses an explicit numerical example and exact bounds on continuation values. The inequalities are strict, and the focal sequence survives in a neighborhood of that example. We do not characterize every equilibrium, and the sufficient conditions include a dynamic verification condition rather than only a list of scalar restrictions on primitives. The numerical example establishes that the mechanism can occur; it is not intended as a calibration of an engineering firm or as an estimate of the duration of reputational protection.

The leading interpretation is organizational: an adviser controls which assignments are released, workers choose how to carry them out, and practical knowledge is passed between overlapping teams. Research supervision offers a related example. A supervisor proposes a demanding problem; a student chooses how intensively to approach it; and the next student can learn a particular method by participating in the work. There are important limits to both interpretations. In particular, we assume that withdrawing a proposal prevents the corresponding task from being undertaken, and that standard implementation is the activity through which the relevant cost-saving knowledge is transmitted. These are restrictions on the setting, not general claims about firms or scientific research.\footnote{The title uses the term \emph{control group} to refer to the available comparison with less intensive implementation. There is no randomized untreated group in the model. The actual state variable is inherited practical knowledge, and the information comes from the outcome of the procedure that the worker chooses.}

\subsection{Related literature}

The question is motivated in part by \citet{PrendergastStole1996}. They show that a manager who wants to appear a fast learner may initially put too much weight on private information and later become reluctant to revise earlier decisions. We instead ask what happens when the person being evaluated delegates implementation to somebody else. The recipient's response determines what the adviser can learn, and what can be learned about her. An important difference is that our reputational-exposure cost uses concave career returns; their mechanism does not rely on such curvature.

Reputation can already discourage experimentation or the disclosure of information. \citet{HalacKremer2020} study experimentation with career concerns, while \citet{CamaraDupuis2023} consider an expert who can avoid evaluation by recommending inaction. \citet{AndinaGarcia2020} and \citet{Hauser2023} study suppression of news and censorship, and \citet{Pei2026} studies the erasure of records observed by later consumers. In the present model, no observation is removed after it is produced. The adviser's current choice changes the cost of the productive action that would generate the next informative observation.

The analysis also relates to endogenous monitoring. \citet{Min2025} studies reputation in an advice relationship where the expert privately knows whether she is informed; the common uncertainty about ability here rules out that particular signaling problem. \citet{Liu2011} and \citet{BarIsaacDeb2021} study reputation with costly observation or opportunities to coast. In \citet{MarinovicSzydlowski2022}, monitoring ability requires maintenance, and an agent's behavior can undermine the incentive to maintain it; \citet{MarinovicSzydlowski2023} study monitoring and transparency. \citet{DebIshii2025} analyze reputation when successive opponents are uncertain about the monitoring structure. Their long-lived player knows a persistent monitoring state; here the adviser shares the public's uncertainty about her own ability, and implementation changes the successor's inherited practical knowledge. \citet{ChenEtAl2026} study successive firms testing a reputational regulator, and \citet{AchimKnoepfle2025} study the interaction of trust and costly oversight. These are close precedents for endogenous exposure, but they do not use the transmission of a cost-saving implementation skill between successive decision makers.

A further comparison concerns the transmission of knowledge itself. \citet{CarlinManso2011} study strategic obfuscation and investor sophistication; our intergenerational cost state makes a different distinction between possessing an outcome record and being able to generate another observation cheaply. \citet{GarfagniniStrulovici2016} show how current experimentation can expand the technologies available to successors. \citet{CheHorner2018} study recommendations that induce successive users to generate information. \citet{Strulovici2022} studies sequential investigation when inherited evidence is subject to attrition. The present paper combines the production of fresh evidence with a separate transmission problem: an outcome remains in the public record, while the practical knowledge needed to obtain further evidence cheaply may be lost. Strategic resistance to knowledge transfer is not itself new; \citet{KuangEtAl2026} study an expert who is reluctant to teach a novice because doing so reduces future rents. Here the incoming worker chooses whether to learn, and the adviser's motive concerns future evaluation rather than competition with a newly trained expert.

Two companion projects provide the immediate background. \emph{Endogenous Vindication} \citep{LukyanovVlasova2026} studies a long-lived expert whose reputation affects a recipient's implementation effort and hence the informativeness of the resulting outcome. \emph{Effort without Evidence} \citep{LukyanovEffort2026} studies successive experimenters who exchange cheap-talk accounts of hidden effort. The present model retains delegated implementation and intergenerational transmission, but changes the latter to hands-on learning. Implementation choices and outcomes are public, and there is no strategic predecessor-message game. This difference is substantive: the equilibrium below is not a cheap-talk equilibrium with an additional long-lived player.

The remainder of the paper is organized as follows. Section~\ref{sec:model} describes the model. Section~\ref{sec:implementation} derives the workers' implementation decisions and explains their incentives to acquire practical knowledge. Section~\ref{sec:equilibrium} introduces the adviser's dynamic problem and states the equilibrium result. Section~\ref{sec:comparison} identifies the role of knowledge loss and studies the subsequent learning process. Section~\ref{sec:conclusion} concludes. The appendices provide the full verification and the exact numerical certificates.

\section{The model}\label{sec:model}

\subsection{The adviser and the projects}

Time is discrete, $t=0,1,\ldots$. A long-lived adviser interacts with one active worker at each date. A worker can first acquire knowledge while observing her predecessor and then undertake one project of her own. She has no payoff from subsequent generations' projects. In the organizational interpretation, a \emph{worker} may stand for a project team whose internal decisions are taken jointly.

The adviser's diagnostic ability is fixed:
\[
\theta\in\{B,G\}.
\]
The labels $G$ and $B$ stand for higher and lower ability. Neither the adviser nor any worker observes $\theta$. If $h_t$ denotes the public history at the beginning of date $t$, their common belief about ability is
\[
p_t=\Prb(\theta=G\mid h_t).
\]
In what follows, we refer to $p_t$ as the adviser's \emph{reputation}. It is a belief about the quality of her assessments, not about whether she is honest.

Each date represents an opportunity to release a project for which the adviser's assessment is that intensive implementation may be necessary. She chooses between releasing the proposal, denoted by $\R$, and withholding it, denoted by $\Sact$. A released proposal includes this warning. The model takes the conditional accuracy of the assessment as a primitive: its probability of being correct is $q_\theta$, where
\[
0<q_B<q_G<1.
\]
The adviser does not observe whether the warning is correct before the project is carried out. In particular, she has no additional private signal of her ability at the recommendation stage.

The restriction to these advisory opportunities is part of the model's clock. A date is not every day on which the adviser might say something, and we do not insert additional dates for unrelated work or other assessments. This matters because the transmission of knowledge takes place between successive opportunities in the model.

A proposal must be released before the worker can undertake the corresponding task. This can represent a project assignment, access to a technical problem, or the provision of a project-specific plan. If the adviser withholds it, the project expires and current project surplus is zero. Thus $\R$ is not an unrestricted cheap-talk sentence that the worker could ignore while undertaking exactly the same task. The adviser controls access to the proposal, but she does not control how the worker implements it.\footnote{A firm in which workers can independently obtain the same assignments after the adviser remains silent would require an additional participation or access decision. We do not claim that the equilibrium below applies to that environment without modification.}

\subsection{Implementation and practical knowledge}

After observing $\R$, the worker chooses between two ways of carrying out the project. With \emph{standard implementation}, denoted by $\M$, she uses moderate resources and accepts the risk that they will be insufficient. With \emph{intensive implementation}, denoted by $\Hact$, she incurs a higher cost but completes the project regardless of whether the warning was correct.

The value of a completed project is normalized to one. Standard implementation succeeds if the warning was wrong and fails if it was correct. Write $y=w$ for the former event and $y=c$ for the latter, so that
\[
\Prb(y=c\mid\theta,\M)=q_\theta.
\]
Consequently, a failure under standard implementation is \emph{favorable evidence about the adviser}: it confirms that her warning was justified. This distinction between project success and assessment accuracy will be important below.

The active worker privately knows whether she has inherited practical knowledge of the standard procedure. Let $K_t\in\{0,1\}$ denote this state. A worker with $K_t=1$ uses the procedure at cost $c_M$. A worker with $K_t=0$ can reconstruct it and obtain the same outcome distribution, but pays an additional cost $b>0$. Intensive implementation costs $c_H$ and does not require this particular knowledge.

Let
\[
q(p)=pq_G+(1-p)q_B.
\]
The worker's expected project surpluses are therefore
\[
U_M(p,1)=1-q(p)-c_M,
\qquad
U_M(p,0)=1-q(p)-c_M-b,
\]
and
\[
U_H=1-c_H.
\]
She selects the higher-surplus implementation, choosing $\Hact$ in case of indifference. Under the conditions used below, $U_H>0$, so an additional option to decline a released project for zero would not alter this decision.

Notice what $K$ does and does not represent. An inexperienced worker is not unable to experiment. She can obtain precisely the same strong evidence by paying $b$. Inherited knowledge instead makes standard implementation cheaper, and this cost difference may determine whether she chooses it. Any subsequent loss of information must therefore be explained by an optimal implementation choice, not by a prohibition on testing.

\subsection{Evidence and the public record}

The worker's implementation choice becomes public after the adviser has released the proposal and before its outcome. The outcome of standard implementation is also public and cannot be concealed. Bayes' rule gives
\[
p^+(p)=\frac{pq_G}{q(p)},\qquad
p^-(p)=\frac{p(1-q_G)}{p(1-q_G)+(1-p)(1-q_B)}.
\]
Here $p^+$ follows $y=c$, a correct warning, and $p^-$ follows $y=w$, an incorrect warning. The superscripts describe the direction of the reputation update, not the success or failure of the project.

Intensive implementation always produces the normalized project output, so that output alone does not reveal whether standard implementation would have sufficed. We nevertheless allow a separate assessment of the warning. With probability $1-\varepsilon$ it produces no finding; otherwise it produces a binary public observation. Specifically,
\[
\Prb(z=1\mid\theta,\Hact)=\varepsilon a_\theta,
\qquad
\Prb(z=0\mid\theta,\Hact)=\varepsilon(1-a_\theta),
\]
where
\[
0<\varepsilon<1,\qquad 0<a_B<a_G<1.
\]
The null observation, $z=\varnothing$, has probability $1-\varepsilon$ under either ability. In the leading example, the non-null observation can be interpreted as an occasional external review of whether the project's initial difficulty assessment was sound. Its occurrence and accuracy are exogenous; no player chooses an auditing policy.

Writing
\[
a(p)=pa_G+(1-p)a_B,
\]
the corresponding updates are
\[
p^1(p)=\frac{pa_G}{a(p)},\qquad
p^0(p)=\frac{p(1-a_G)}{p(1-a_G)+(1-p)(1-a_B)}.
\]
A null observation leaves $p$ unchanged. So does withholding a proposal, since the adviser does not know her ability and her action at a given history has the same likelihood under both types.

All outcome draws are fresh: conditional on $\theta$ and the implementation chosen at the current date, the outcome law above is independent of the past, of the workers' knowledge states, and of their learning costs. Equivalently, Nature uses independent fresh innovations across dates, while endogenous implementation selects which of the two experiments is observed. This conditional formulation does not assert that the observations generated by an adaptive equilibrium policy are unconditionally independent.

Both forms of evidence remain in the public archive. We call standard implementation the strong test and the occasional assessment the weak test. No general Blackwell ranking of the two non-null binary signals is imposed by the notation; at the numerical example, the small arrival probability makes the intensive-implementation experiment much less informative per date.

\subsection{Learning from a predecessor}

If the current worker chooses standard implementation, the incoming worker can participate in a live demonstration. Before its outcome is realized, she privately observes a learning cost
\[
\kappa_t\sim U[0,\bar\kappa]
\]
and decides whether to pay it. Costs are iid and independent of $\theta$ and of all outcome innovations. Payment gives her $K_{t+1}=1$; otherwise $K_{t+1}=0$. She evaluates this cost against the expected saving on her own subsequent project, with no additional discounting over these two stages of her life.

There is no such opportunity after intensive implementation or after a withheld proposal. In those cases $K_{t+1}=0$. A worker who has the knowledge leaves after her own project, so her knowledge does not remain available as an independent stock. A worker who reconstructs the procedure can, however, provide a live demonstration just as an experienced worker can. Thus knowledge can be rebuilt after its transmission has been interrupted.

The learning technology is specific. It concerns practical familiarity with the standard procedure, not general education or the ability to undertake intensive work. The incoming worker must see the procedure being used, and reading the public outcome archive alone does not remove the reconstruction cost. We assume that the outgoing worker's ordinary work provides this opportunity without an additional teaching decision. These restrictions give a precise meaning to the intergenerational link.

Learning is private. Let
\[
x_t=\Prb(K_t=1\mid h_t)
\]
be the public probability that the current worker has inherited the knowledge. The adviser knows $x_t$ but not $K_t$. In equilibrium $x_t$ will often equal zero or one, although interior values are needed to describe every history. A single worker's unobserved decision not to learn does not change the public conjecture $x_{t+1}$. She therefore cannot induce a more favorable recommendation by secretly declining the training.

\subsection{Preferences, timing, and equilibrium}

The adviser values the current project's surplus and her reputation. Her objective is
\[
\E\sum_{t=0}^{\infty}\delta^t
\left[u_t+\beta r(p_{t+1})\right],
\qquad
r(p)=2p-p^2,\qquad 0<\delta<1,\quad\beta>0.
\]
Here $u_t$ is the active worker's realized project return net of implementation and, when applicable, reconstruction costs. The expected values of this component are the $U_M$ and $U_H$ specified above. The incoming worker's separate learning cost is not included in $u_t$. The adviser's objective is therefore not aggregate social welfare.

The career return $r$ is increasing and concave. Favorable public assessments are valuable, but their marginal value diminishes. This can represent diminishing returns to reputation in future demand or status, or risk aversion over a career reward that rises with perceived ability. We use the quadratic form to obtain an exact expression for the cost of reputational exposure. There are no transfers, employment contracts, or payments conditional on learning or outcomes.

For clarity, the sequence within a date is as follows.
\begin{enumerate}[label=\arabic*.,leftmargin=2em]
\item The public history determines $(p_t,x_t)$. The active worker privately knows $K_t$.
\item The adviser releases the proposal and warning, $\R$, or withholds it, $\Sact$.
\item Following $\R$, the worker chooses $\M$ or $\Hact$. The choice becomes public; the adviser cannot then withdraw her proposal.
\item Following $\M$, the incoming worker observes her private cost and decides whether to learn, before the implementation outcome is known.
\item The outcome or occasional assessment becomes public, current payoffs are realized, and beliefs about ability are updated. The incoming worker becomes active at the next date.
\end{enumerate}

The public state contains two different kinds of uncertainty: $p$ concerns the adviser, whereas $x$ concerns the current worker. They need not be collapsed into one belief. The timing allows them to be treated separately.

\begin{lemma}\label{lem:factorization}
Suppose $\theta$ and $K_0$ are initially independent. Under the stated information structure, beliefs generated by behavioral strategies preserve
\[
\Prb(\theta,K\mid h)=\Prb(\theta\mid h)\Prb(K\mid h).
\]
The same factorization holds under completely mixed behavioral perturbations.
\end{lemma}
\begin{proof}
At a given history, the adviser has no private information about $\theta$. The active worker's implementation may depend on $K$, but not on additional information about $\theta$. Observing implementation therefore updates the distribution of $K$ without updating $p$. Conditional on that implementation, the outcome depends on $\theta$ but not on $K$. Finally, the incoming worker decides whether to learn before seeing that outcome, using a cost independent of $\theta$. Applying these observations successively proves the claim.
\end{proof}

We study stationary sequential equilibria in which strategies depend on public history through $(p,x)$, with a worker's own knowledge or learning cost included at her private decision nodes. Sequential rationality is required at every history, and beliefs must be limits of Bayes-consistent beliefs under completely mixed perturbations. We refer to this as a \emph{factorized Markov stationary sequential equilibrium}. We use this class to establish existence, not to assert that every equilibrium of the game must be stationary.

\section{Implementation and learning}\label{sec:implementation}

Before studying the adviser's recommendation, let us determine what a worker does when a proposal is released. This comparison explains why inherited knowledge matters for the adviser's evaluation.

\subsection{When does knowledge change implementation?}

Write
\[
d=q_G-q_B>0,\qquad D=c_H-c_M.
\]
The advantage of standard over intensive implementation is
\[
U_M(p,0)-U_H=D-b-q_B-dp
\]
for an inexperienced worker and
\[
U_M(p,1)-U_H=D-q_B-dp
\]
for an experienced worker. Both advantages decrease with reputation. The more credible the warning that standard resources will be insufficient, the less attractive it is to economize on those resources.

Suppose
\[
q_B+b<D<q_G.
\]
There are then two interior thresholds,
\[
p_R=\frac{D-b-q_B}{d},\qquad
p_T=\frac{D-q_B}{d},\qquad p_R<p_T.
\]
The subscript $R$ refers to a worker who must reconstruct the procedure, and $T$ to one who has been trained.

\begin{proposition}\label{prop:implementation}
Conditional on a released proposal, an inexperienced worker chooses $\M$ if and only if $p<p_R$, and an experienced worker chooses $\M$ if and only if $p<p_T$.
\end{proposition}
\begin{proof}
Each worker chooses standard implementation precisely when the corresponding advantage above is positive. At equality, the specified tie-breaking rule selects intensive implementation.
\end{proof}

There are consequently three regions. Below $p_R$, even paying to reconstruct the procedure is worthwhile. Above $p_T$, even an experienced worker prefers intensive implementation. Between the two thresholds, inherited knowledge changes the action: $\M$ when $K=1$, but $\Hact$ when $K=0$.

We shall call $[p_R,p_T)$ the \emph{knowledge-sensitive region}. Its width is
\[
p_T-p_R=\frac{b}{q_G-q_B}.
\]
The public belief about the adviser may be identical for two successive workers, yet their optimal choices differ because one has learned from her predecessor and the other has not. This is the comparison that will make a break in knowledge transmission valuable to the adviser.

\subsection{Why acquiring knowledge can be worthwhile}

The incoming worker does not learn in order to evaluate the adviser for society. She learns because knowledge can reduce the cost of her own implementation. But she must pay before observing the current outcome, and that outcome may change both the next recommendation and her own implementation decision. Her willingness to learn therefore cannot be established from the current cost saving alone.

The following condition gives a useful way to resolve this problem:
\[
(D-q_B)(q_G-D)>b(1-D),
\qquad
\bar\kappa<b(1-D).
\]
The first inequality is equivalent to $p^-(p_T)<p_R$. Thus, following any rational use of $\M$, an outcome unfavorable to the adviser sends reputation below the inexperienced worker's threshold. In the equilibrium constructed below, the adviser releases a project at every such low belief, whatever the public probability of training. On that branch, the incoming worker will use $\M$ whether or not she has learned, but learning saves $b$.

The probability of that branch is $1-q(p)>1-D$. Her expected private saving is therefore greater than $b(1-D)$, which exceeds even the highest learning cost $\bar\kappa$. This makes training strictly optimal for every cost type. Notice that this argument uses the adviser's low-reputation policy. It is part of the joint equilibrium verification, not a claim about learning under arbitrary future recommendations.

At other outcomes the adviser may withhold the next project, making the acquired knowledge useless for that worker's own assignment. This does not invalidate the learning decision: the possibility of the low-reputation branch already pays for it. Because learning is private, an individual worker who declines to learn cannot make the adviser change her recommendation by announcing a lower public training probability.

It follows that, along the proposed equilibrium and after its payoff-relevant deviations, $\M$ sends the next state to $x'=1$. Both $\Hact$ and $\Sact$ send it to $x'=0$. The appendix checks learning incentives also after implementation choices that occur only as trembles.

\subsection{Why an informative outcome is costly to the adviser}

An increase in the probability of a correct warning is not the same thing as an increase in expected reputation following a test. Bayesian reputation is a martingale:
\[
\E[p'\mid p]=p.
\]
The adviser cannot expect a favorable revision merely because she agrees to be evaluated. With concave career returns, she may instead dislike the uncertainty that evaluation creates.

For our quadratic specification,
\[
r(p)-\E[r(p')\mid p]=\operatorname{Var}(p'\mid p).
\]
For standard implementation this variance is
\[
J_M(p)=\frac{p^2(1-p)^2d^2}{q(p)[1-q(p)]}.
\]
For a non-null external assessment it is
\[
J_H(p)=\frac{p^2(1-p)^2(a_G-a_B)^2}{a(p)[1-a(p)]},
\]
so the variance per date of intensive implementation is $\varepsilon J_H(p)$.

If the adviser cared only about the current date, releasing a proposal that induces standard implementation by an experienced worker would give, relative to withholding it,
\[
U_M(p,1)-\beta J_M(p).
\]
The example below makes this expression strictly positive at the history where the long-lived adviser nevertheless withholds the proposal. Thus the interruption cannot be attributed simply to a negative one-period return. What matters is the continuation after each choice.

\section{Reputation and interruptions of advice}\label{sec:equilibrium}

\subsection{The continuation problem}

Under full renewal after standard implementation, the next knowledge state is either zero or one. Let $V_k(p)$ be the adviser's equilibrium continuation value when the worker's knowledge state is publicly known to be $k\in\{0,1\}$. For a bounded function $f$, define
\[
\cP_Mf(p)=q(p)f(p^+(p))+[1-q(p)]f(p^-(p))
\]
and
\[
\cP_Hf(p)=(1-\varepsilon)f(p)
+\varepsilon\{a(p)f(p^1(p))+[1-a(p)]f(p^0(p))\}.
\]
These operators take expectations over the next reputation under the two implementation choices.

The value of withholding a proposal is
\[
Q^S(p)=\beta r(p)+\delta V_0(p).
\]
Current project surplus is zero; reputation is unchanged; and the next worker has no inherited knowledge. If releasing the proposal induces standard implementation, its value is
\[
Q_k^M(p)=U_M(p,k)+\cP_M[\beta r+\delta V_1](p).
\]
Here the current outcome updates reputation, and the live demonstration allows the successor to acquire knowledge. If it induces intensive implementation, the value is
\[
Q^H(p)=U_H+\cP_H[\beta r+\delta V_0](p).
\]
These are induced action values. The adviser chooses whether to release the proposal, while the worker chooses its implementation.

Consider the following candidate recommendation policy. With an inexperienced worker, the adviser always releases the proposal. With an experienced worker, she releases it outside the knowledge-sensitive region and withholds it inside that region. Combining this policy with Proposition~\ref{prop:implementation} gives the following outcomes:
\begin{center}
\begin{tabular}{lll}
\toprule
Reputation & Inexperienced worker & Experienced worker\\
\midrule
$p<p_R$ & Standard, with reconstruction & Standard\\
$p_R\le p<p_T$ & Intensive & Proposal withheld\\
$p\ge p_T$ & Intensive & Intensive\\
\bottomrule
\end{tabular}
\end{center}
The table describes the equilibrium outcome, not the worker's response to a proposal that was withheld. In its middle-right entry, an experienced worker would choose standard implementation if the adviser deviated and released the proposal.

\subsection{A path through the equilibrium}

Let us first describe the mechanism using an example; the existence result follows in the next subsection. Take
\begin{equation}
\begin{gathered}
q_B=\frac35,\quad q_G=\frac34,\quad c_M=\frac1{10},\quad
b=\frac1{100},\quad c_H=\frac{19}{25},\\
\beta=14,\quad\delta=\frac7{10},\quad\varepsilon=\frac1{100},\quad
a_B=\frac25,\quad a_G=\frac35,\quad\bar\kappa=\frac1{400}.
\end{gathered}
\tag{W}\label{eq:W}
\end{equation}
The two worker thresholds are
\[
p_R=\frac13,\qquad p_T=\frac25.
\]
Start with an experienced worker and reputation $p_0=5/16$, which is below $p_R$. The adviser releases the proposal, the worker uses standard implementation, and the successor acquires knowledge.

Suppose that the warning is confirmed. The project has failed under standard implementation, but this improves the adviser's reputation:
\[
p_1=p^+(p_0)=\frac{25}{69}.
\]
This new belief lies strictly between $1/3$ and $2/5$. The next worker, who has inherited knowledge, would still use standard implementation. The adviser instead withholds the proposal.

There is no outcome during that interruption, so reputation stays at $p_1$. But there is also no live demonstration: the following worker starts without the practical knowledge. At the same public belief she now prefers intensive implementation. The adviser releases the next proposal, and the project proceeds.

The choice is quantitatively small on the worker's side. At $p_1$,
\[
U_M(p_1,1)=\frac{113}{460}\simeq0.24565,
\qquad
U_H=\frac6{25}=0.24,
\]
whereas the inexperienced worker's standard-implementation surplus is lower by $b=0.01$, and hence below $0.24$. The loss of a small cost advantage changes the implementation decision. It does not make the worker more trusting, and it does not make the procedure impossible to reconstruct.

The adviser's preference also changes for a specifically dynamic reason. If one postponement left the next worker experienced, releasing the proposal at $p_1$ would be better than postponing it. With actual knowledge loss, the opposite inequality holds. Section~\ref{sec:comparison} states this comparison precisely.

The confirmation leading to $p_1$ occurs with positive probability under either ability. If instead the first warning is disconfirmed, reputation falls to $p^-(p_0)=25/113$, below $p_R$, and standard implementation remains attractive. The theorem does not say that every realization follows the interruption path.

\subsection{Existence and what supports it}

Four economic requirements organize the construction. First, reconstruction must be costly enough to create two distinct, interior implementation thresholds, with an initial reputation that can cross into the knowledge-sensitive region after a confirming outcome. Second, the learning cost must be small enough that the disconfirming branch makes acquiring knowledge worthwhile. Third, the project at $p_1$ must remain worthwhile even after charging the one-period cost of reputational exposure. Finally, releasing a project must be profitable when it induces reconstruction at low beliefs or intensive implementation at higher beliefs.

These requirements are not sufficient on their own. The adviser compares infinite-horizon continuation values, and we need to establish that withholding is optimal throughout the knowledge-sensitive region. Appendix~\ref{app:verification} states the four restrictions as Conditions~\ref{cond:c1}--\ref{cond:c4}, followed by Condition~\ref{cond:c5}, which bounds the remaining continuation-value comparisons. That last condition is an explicit verification certificate, not another scalar assumption with an immediate economic interpretation.

\begin{theorem}\label{thm:main}
Under the model of Section~\ref{sec:model} and Conditions~\ref{cond:c1}--\ref{cond:c5} in Appendix~\ref{app:verification}, the game has a factorized Markov stationary sequential equilibrium with the displayed recommendation policy. The incentive and crossing inequalities in those conditions are strict; the subsolution and supersolution inequalities are weak.

In the knowledge-sensitive region, there is a unique cutoff $\widehat x(p)\in(0,1)$ in the public probability of an experienced worker: the adviser releases the proposal below this cutoff and withholds it above the cutoff. Outside that region she releases it at every $x$.

Starting from $(p_0,1)$, a proposal induces standard implementation and every incoming worker learns. A confirming outcome reaches $p_1$, where the adviser withholds the next proposal. Inherited knowledge is lost, after which she resumes proposals with intensive implementation and occasional informative assessments. At $p_1$, preserving knowledge after one postponement reverses the adviser's strict preference for withholding.

The parameters in \eqref{eq:W} satisfy the conditions. The focal path and the strict reversal persist on an open neighborhood in the sense of Lemma~\ref{lem:robust}. Under low ability, the subsequent phase of intensive implementation ends almost surely.
\end{theorem}

The complete proof is in Appendix~\ref{app:verification}. Its order is useful to keep in mind. We first subtract the discounted career return from a reputation that would never change. This turns the reputational term into a cost proportional to posterior variance. We then bound the candidate policy's continuation values and verify the adviser's deviations. The low-reputation policy and the workers' learning decisions support each other, and are verified jointly. Finally, we complete learning decisions after off-path implementation choices and show that the strict focal comparisons survive small changes in primitives.

The existence result does not rest on a simulated belief grid. The example uses rational parameters, and the inequalities are checked using exact algebra and finite recursions. This is useful for establishing existence but should not obscure the result's scope: we have constructed an equilibrium with the stated behavior, not characterized all equilibria of a general class of organizations.

\section{The role of knowledge loss}\label{sec:comparison}

\subsection{Withholding or merely postponing?}

An adviser with career concerns might prefer not to be tested even if knowledge were never lost. To distinguish this familiar motive from the mechanism here, hold fixed the continuation values of the original equilibrium and change only what one current postponement does to the next worker.

For $p$ in the knowledge-sensitive region, write $Q_D^S(p)=Q^S(p)$ for actual withholding and define the standard-implementation margin
\[
A_D(p)=Q_1^M(p)-Q_D^S(p).
\]
The subscript $D$ refers to the loss of inherited knowledge. Consider instead a one-period postponement that preserves $K=1$, after which the original continuation game resumes. Its corresponding margin is
\[
A_K(p)=Q_1^M(p)-[\beta r(p)+\delta V_1(p)].
\]
The project benefit and the current experiment are identical in these comparisons. Only the successor's knowledge after postponement is different.

Under the constructed policy, define
\[
G_0(p)=Q^H(p)-Q_D^S(p)=V_0(p)-V_1(p)>0.
\]
This is the adviser's advantage, at that belief, from facing an inexperienced rather than an experienced worker. The former implements intensively, whereas the latter would choose the informative standard procedure if a proposal were released.

\Needspace{8\baselineskip}
\begin{proposition}\label{prop:local}
Withholding is strictly preferred under actual knowledge loss, but releasing the proposal is strictly preferred to a one-period knowledge-preserving postponement, if and only if
\[
-\delta G_0(p)<A_D(p)<0.
\]
\end{proposition}
\begin{proof}
Subtracting the two margins gives
\[
A_K(p)-A_D(p)
=\delta[V_0(p)-V_1(p)]
=\delta G_0(p).
\]
Thus $A_D(p)<0<A_K(p)$ is equivalent to the stated inequalities.
\end{proof}

The right inequality says that the adviser actually withholds the proposal. The left says that the gain from changing the successor's implementation is large enough to account for this decision: remove that gain, and she would put the project forward. This is an exact necessary-and-sufficient comparison, conditional on the equilibrium continuation values. It is not a primitive characterization of when such an equilibrium exists.

At $p_1$ in example \eqref{eq:W}, the exact bounds imply
\[
A_D(p_1)<-0.01843,\qquad A_K(p_1)>0.00340.
\]
Moreover, the one-period margin remains positive:
\[
U_M(p_1,1)-14J_M(p_1)
=\frac{111524267}{651019140}>\frac{17}{100}.
\]
The adviser is therefore not withholding an intrinsically unprofitable project. She is forgoing a current gain in order to affect how subsequent projects will be implemented.

\subsection{Preserving knowledge after every interruption}

The preceding comparison changes only one transition. A stronger exercise changes the technology permanently and solves the adviser's problem again. Suppose that withholding preserves the current knowledge state, standard implementation gives the successor knowledge, and intensive implementation does not. On the states $k\in\{0,1\}$, the imposed transition rule is
\[
\Sact:k\mapsto k,\qquad
\M:k\mapsto1,\qquad
\Hact:k\mapsto0.
\]
Worker implementation still follows the same payoff comparison, while the adviser optimizes given the altered transitions.

This is a mechanical technology counterfactual. It does not assert that the original private learning decisions would remain unchanged if workers were offered a new way to preserve their knowledge. In particular, it is not a welfare evaluation of a training subsidy or a proposed institutional reform.

\begin{corollary}\label{cor:preserve}
At \eqref{eq:W}, the adviser in the knowledge-preserving economy releases the proposal at both $(p_0,1)$ and $(p_1,1)$. The interruption at $p_1$ is absent.
\end{corollary}
\begin{proof}
Let $\widetilde{\cB}$ be the new optimal Bellman operator. Zero and the constant $\overline V=143/3$ are respectively a subsolution and a supersolution, so
\[
\widetilde{\cB}^{18}0\le\widetilde V
\le\widetilde{\cB}^{18}\overline V.
\]
Using the lower bound for the continuation following a released proposal and the upper bound for the continuation following postponement, exact recursion gives
\[
\widetilde Q_1^M(p_0;\widetilde{\cB}^{18}0)
-\widetilde Q_1^S(p_0;\widetilde{\cB}^{18}\overline V)
>\frac1{20}
\]
and
\[
\widetilde Q_1^M(p_1;\widetilde{\cB}^{18}0)
-\widetilde Q_1^S(p_1;\widetilde{\cB}^{18}\overline V)
>\frac2{25}.
\]
Both comparisons are strict even under these conservative continuation bounds. Appendix~\ref{app:preservation} specifies the operator and the exact certificate.
\end{proof}

Thus the result does not depend on evaluating preservation while leaving the adviser's old policy in place. Even after she is allowed to adjust her future decisions, preserving knowledge eliminates the focal withholding decision in the example.

\subsection{How long can the adviser remain weakly tested?}

Once the project resumes with intensive implementation, most dates produce no new information about ability. But the occasional assessment remains informative. A low-ability adviser therefore cannot rely on this protection forever.

\begin{corollary}\label{cor:meta}
Under $\theta=B$ at \eqref{eq:W}, the intensive-implementation phase following the interruption ends almost surely. Starting at reputation $p_1$, its expected duration is $500$ advisory opportunities.
\end{corollary}
\begin{proof}
Conditional on a non-null assessment and $\theta=B$, posterior odds are multiplied by $3/2$ with probability $2/5$ and by $2/3$ with probability $3/5$. The starting odds are $25/44$; the threshold odds corresponding to $p_R=1/3$ are $1/2$. If a favorable finding is counted as $+1$ and an unfavorable finding as $-1$, the phase ends when this integer random walk first reaches $-1$.

The expected number of non-null observations needed is
\[
\frac1{3/5-2/5}=5.
\]
Their independent waiting times have mean $1/\varepsilon=100$, giving $5/\varepsilon=500$ dates. The walk has negative drift and reaches $-1$ almost surely.
\end{proof}

The number $500$ is a property of the example, not an empirical prediction. The important distinction is between a long delay and permanent nonidentification. When enough unfavorable assessments push reputation below $p_R$, even an inexperienced worker reconstructs the standard procedure. Informative implementation and the opportunity for knowledge transfer then return.

\section{Discussion and conclusion}\label{sec:conclusion}

The model describes a particular interaction between reputation and practical knowledge. A worker who knows how to carry out a task cheaply may be more willing to use a procedure that also tests the adviser's judgment. Carrying out the task gives the next worker an opportunity to acquire the same knowledge. The adviser can therefore affect future evaluation by deciding whether to release a project during the interval in which the two generations overlap.

The result is not that practical knowledge is generally bad for an adviser. At low reputations, knowledge simply saves reconstruction costs, and the adviser releases the project whether the worker is experienced or not. At high reputations, both kinds of worker implement intensively. The conflict arises in between, where knowledge changes the worker's action. At the relevant history, the adviser would accept standard implementation if postponement preserved the successor's knowledge, but withholds the proposal when postponement removes it.

Three features of the model deserve emphasis when applying this reasoning. First, control over the proposal is essential: the worker cannot independently undertake the same task after it is withheld. Second, the intergenerational link concerns hands-on acquisition of a cost-saving procedure. It is not strategic word-of-mouth communication about hidden effort. Third, the public record is complete for the actions and outcomes specified in the model. What disappears during an interruption is an inexpensive way of producing another informative outcome, not the evidence already accumulated.

These qualifications help locate the organizational application. A technical adviser working with rotating teams can be evaluated through their ordinary work, and that work can also transmit knowledge to the next team. Research supervision provides a similar interpretation when learning a method requires participation in a live project. Neither interpretation requires that every scientific outcome be experimentally controlled, that a supervisor know her own ability, or that successors accept a predecessor's unsupported causal claims. Nor do the examples establish that supervisors or advisers usually exploit this possibility; they identify the setting in which the incentive can arise.

The theorem remains an existence result with a verified example and local robustness. It does not establish uniqueness, rank organizational arrangements by welfare, or provide a general characterization using only scalar primitive restrictions. In particular, the knowledge-preservation exercise changes a transition technology mechanically, rather than deriving a feasible policy that implements it through private training decisions. A richer theory of organizational design would have to address that distinction.

The central conclusion is that retaining a record and retaining the means to add to it are different tasks. In our example, the record is public throughout. Nonetheless, a temporary interruption changes the next worker's implementation choice and allows the adviser to resume working under much weaker evaluation. Under low ability, the phase of weak evaluation eventually ends, but only after a period in which the practical knowledge that would have made a stronger comparison attractive has ceased to be passed on.

\clearpage
\appendix
\section{Equilibrium verification}\label{app:verification}

This appendix supplies the conditions and full proof used in Theorem~\ref{thm:main}. We retain the notation $k=0,1$ for publicly known knowledge states. A \emph{face} below means one of these two endpoints of the public probability $x$, not an additional state variable. The proof evaluates a proposed policy, verifies the adviser and workers' incentives jointly, and then completes histories with interior $x$.

\subsection{Transformed values and the candidate policy}

Let
\[
d=q_G-q_B>0,\qquad D=c_H-c_M,
\]
and define
\[
p_R=\frac{D-b-q_B}{d},\qquad
p_T=\frac{D-q_B}{d}.
\]
An untrained worker uses standard implementation below $p_R$; a trained worker uses standard implementation below $p_T$.

For a bounded function $f$, set
\[
\cP_Mf(p)=q(p)f(p^+(p))+[1-q(p)]f(p^-(p)),
\]
\[
\cP_Hf(p)=(1-\varepsilon)f(p)
+\varepsilon\{a(p)f(p^1(p))+[1-a(p)]f(p^0(p))\}.
\]
The posterior-variance losses generated by a standard implementation and a non-null weak trace are
\[
J_M(p)=\frac{p^2(1-p)^2d^2}{q(p)[1-q(p)]},
\qquad
J_H(p)=\frac{p^2(1-p)^2(a_G-a_B)^2}{a(p)[1-a(p)]}.
\]
Let
\[
\gamma=\frac{\beta}{1-\delta},\qquad
m_k(p)=U_M(p,k)-\gamma J_M(p),
\qquad
h(p)=U_H-\gamma\varepsilon J_H(p).
\]
Subtracting the annuity value of an unchanged reputation leaves these implementation returns net of capitalized exposure costs. A negative $m_1$ is not a negative current project payoff; it includes the future career cost of the posterior change.

For reference, the candidate policy on the publicly known knowledge states is
\begin{center}
\begin{tabular}{lll}
\toprule
Face & Belief & Action induced by the adviser\\
\midrule
$x=0$ & $p<p_R$ & $\R,\M$ with rebuild\\
$x=0$ & $p\ge p_R$ & $\R,\Hact$\\
$x=1$ & $p<p_R$ & $\R,\M$\\
$x=1$ & $p_R\le p<p_T$ & $\Sact$\\
$x=1$ & $p\ge p_T$ & $\R,\Hact$\\
\bottomrule
\end{tabular}
\end{center}
For transformed values $W_k=V_k-\gamma r$, its policy operator is
\[
(\cT^\pi W)_0(p)=
\begin{cases}
m_0(p)+\delta\cP_MW_1(p),&p<p_R,\\
h(p)+\delta\cP_HW_0(p),&p\ge p_R,
\end{cases}
\]
\[
(\cT^\pi W)_1(p)=
\begin{cases}
m_1(p)+\delta\cP_MW_1(p),&p<p_R,\\
\delta W_0(p),&p_R\le p<p_T,\\
h(p)+\delta\cP_HW_0(p),&p\ge p_T.
\end{cases}
\]

\subsection{Five condition blocks}

Conditions C1--C4 collect the economic restrictions discussed in Section~\ref{sec:equilibrium}. Condition C5 provides explicit bounds for the remaining dynamic comparisons. We keep its distinction between weak comparison-function inequalities and strict incentive inequalities.

\begin{condition}[Implementation thresholds and a confirming outcome]
\label{cond:c1}
\[
q_B+b<D<q_G,
\]
and there is $p_0<p_R$ such that
\[
p_1:=p^+(p_0)\in(p_R,p_T).
\]
\end{condition}

\begin{condition}[Incentives to acquire practical knowledge]
\label{cond:c2}
\[
(D-q_B)(q_G-D)>b(1-D),
\qquad
\bar\kappa<b(1-D).
\]
\end{condition}

\begin{condition}[Positive current value and capitalized exposure]
\label{cond:c3}
At $p_1$,
\[
m_1(p_1)<0<U_M(p_1,1)-\beta J_M(p_1).
\]
Equivalently,
\[
(1-\delta)\frac{U_M(p_1,1)}{J_M(p_1)}
<\beta<
\frac{U_M(p_1,1)}{J_M(p_1)}.
\]
\end{condition}

\begin{condition}[Profitable implementation outside the interruption region]
\label{cond:c4}
\[
\inf_{p<p_R}m_0(p)>0,
\qquad
\inf_{p\ge p_R}h(p)>0,
\]
and $0<\varepsilon<1$, $0<a_B<a_G<1$.
\end{condition}

\begin{condition}[Bounds on continuation values]
\label{cond:c5}
There are explicit bounded vector functions $\ell^D=(\ell^D_0,\ell^D_1)$ and $u^D=(u^D_0,u^D_1)$, constructed from primitives, such that
\[
\ell^D\le\cT^\pi\ell^D,
\qquad
\cT^\pi u^D\le u^D,
\]
\[
\sup_{p\in[p_R,p_T)}
\{m_1(p)+\delta\cP_Mu^D_1(p)-\delta\ell^D_0(p)\}<0.
\]
In addition, an explicit finite continuation plan and/or a primitive subsolution supplies a lower function $f_1$ at the two successors of the focal standard implementation, while an explicit supersolution $u^K$ supplies an upper bound at $p_1$. The certificate must prove these bounds from primitive recursions and satisfy
\[
m_1(p_1)+\delta\cP_Mf_1(p_1)-\delta u^K_1(p_1)>0.
\]
\end{condition}

The first strict inequality checks that releasing a proposal is worse than actual withholding throughout the knowledge-sensitive region. The second compares a released proposal at $p_1$ with a postponement that preserves knowledge for one period. Its finite feasible plan is evaluated only after the candidate policy has been verified. This order avoids using the desired conclusion as a bound on an unknown continuation value.

\subsection{Implementation and learning: joint verification}

\begin{lemma}
The worker's optimal implementation has the thresholds stated in Proposition~\ref{prop:implementation}. At \eqref{eq:W}, they equal $1/3$ and $2/5$.
\end{lemma}

\begin{proof}
For completeness, the two implementation advantages are
\[
U_M(p,0)-U_H=D-b-q_B-dp,
\qquad
U_M(p,1)-U_H=D-q_B-dp.
\]
Condition~\ref{cond:c1} puts both zeros strictly inside $(0,1)$ and orders them as $p_R<p_T$. Under \eqref{eq:W}, these differences become
\[
\frac1{20}-\frac3{20}p,
\qquad
\frac3{50}-\frac3{20}p.
\]
\end{proof}

\begin{lemma}
\label{lem:renewal}
Suppose the adviser follows the proposed policy, including release at every belief below $p_R$ for every public knowledge probability. Under Condition~\ref{cond:c2}, every use of $\M$ at $p<p_T$ makes learning strictly optimal for every incoming cost type. Thus the next public knowledge state is $x'=1$ after those implementations and $x'=0$ after $\Sact$ or $\Hact$.
\end{lemma}

\begin{proof}
The adverse reputation update crosses the lower threshold precisely when
\[
p^-(p_T)<p_R
\quad\Longleftrightarrow\quad
(D-q_B)(q_G-D)>b(1-D).
\]
Because $q(p_T)=D$, the probability of the unfavorable branch is uniformly above $1-D$ for every rational standard implementation. On that branch the successor belief lies below $p_R$, the adviser warns for every competence state, and competence saves rebuild cost $b$. Since a learning deviation is unobserved, the incoming worker evaluates it holding the public training probability and the adviser's policy fixed. The gross private return to training is therefore strictly greater than $b(1-D)>\bar\kappa$. All types train.

At \eqref{eq:W},
\[
p^-(p_T)=\frac5{17}<\frac13,
\qquad
b(1-D)-\bar\kappa
=\frac{17}{5000}-\frac1{400}
=\frac9{10000}>0.
\]
\end{proof}

\subsection{The martingale transform and policy verification}

The three original action values, repeated here to make the verification self-contained, are
\[
Q^S(p)=\beta r(p)+\delta V_0(p),
\]
\[
Q_k^M(p)=U_M(p,k)+\cP_M[\beta r+\delta V_1](p),
\]
\[
Q^H(p)=U_H+\cP_H[\beta r+\delta V_0](p).
\]

\begin{lemma}
For $W_k=V_k-\gamma r$, subtracting the career annuity gives the operator $\cT^\pi$ displayed above. It preserves the pointwise order and is a contraction of modulus $\delta$.
\end{lemma}

\begin{proof}
Posterior belief is a martingale. Since $r(p)=2p-p^2$,
\[
r(p)-\E[r(p')\mid p]=\operatorname{Var}(p'\mid p).
\]
Using $\beta+\delta\gamma=\gamma$, the current career return and the subtracted continuation annuity combine into minus $\gamma$ times the posterior variance. A standard implementation has variance $J_M$; a weak audit has variance $\varepsilon J_H$ because the null trace leaves belief unchanged. Substitution gives $m_k$ and $h$. The remaining claims follow because each continuation is an expectation multiplied by $\delta$.
\end{proof}

\begin{lemma}
Given the learning responses just specified, the proposed recommendations are strict best responses at $x=0,1$. In the knowledge-sensitive region, the release margin is affine in $x$ and has exactly one zero.
\end{lemma}

\begin{proof}
Let $W^\pi$ be the unique fixed point of $\cT^\pi$. Monotonicity and the suppression part of Condition~\ref{cond:c5} give
\[
\ell^D\le W^\pi\le u^D.
\]
Condition~\ref{cond:c4} makes every used return on the $x=0$ face positive, hence $W_0^\pi>0$. The policy equations imply
\[
W_1^\pi(p)=
\begin{cases}
W_0^\pi(p)+b,&p<p_R,\\
\delta W_0^\pi(p),&p_R\le p<p_T,\\
W_0^\pi(p),&p\ge p_T.
\end{cases}
\]
On the untrained face, warning rather than silence has transformed margin $(1-\delta)W_0^\pi(p)>0$. The trained margin adds $b$ below $p_R$ and is identical to the untrained margin at and above $p_T$. In the knowledge-sensitive region, Condition~\ref{cond:c5}'s uniform inequality makes the deviation inducing standard implementation strictly worse than silence. The one-step-deviation principle verifies the face policy.

For interior $x$ in the band, the warning margin is
\[
A(p,x)=(1-x)A_0^H(p)+xA_1^M(p),
\]
where $A_0^H>0>A_1^M$. Thus
\[
\widehat x(p)=\frac{A_0^H(p)}{A_0^H(p)-A_1^M(p)}\in(0,1)
\]
is the unique cutoff. Outside the band both endpoint margins are positive.
\end{proof}

\subsection{Verification of the numerical example}

Under \eqref{eq:W},
\[
J_M(p)=\frac{3p^2(1-p)^2}{(4+p)(8-3p)},
\qquad
J_H(p)=\frac{p^2(1-p)^2}{(2+p)(3-p)},
\]
and
\[
m_1(p)=\frac3{10}-\frac3{20}p-\frac{140}{3}J_M(p),
\quad m_0(p)=m_1(p)-\frac1{100},
\quad h(p)=\frac6{25}-\frac7{15}J_H(p).
\]

\begin{proposition}
At \eqref{eq:W}, Conditions~\ref{cond:c1}--\ref{cond:c5} hold, with strict margins for the incentive and crossing inequalities. Consequently the proposed policy and learning decisions form the equilibrium of Theorem~\ref{thm:main}.
\end{proposition}

\begin{proof}
Condition~\ref{cond:c1} follows from
\[
\frac{61}{100}=q_B+b<\frac{33}{50}=D<q_G=\frac34
\]
and the displayed values of $p_0,p_1,p^-(p_0)$. Condition~\ref{cond:c2} follows from
\[
(D-q_B)(q_G-D)=\frac{27}{5000}>
\frac{17}{5000}=b(1-D)
\]
and the training gap $9/10000$. At $p_1$,
\[
m_1(p_1)<0,
\qquad
U_M(p_1,1)-14J_M(p_1)
=\frac{111524267}{651019140}>\frac{17}{100},
\]
which verifies Condition~\ref{cond:c3}. Direct differentiation gives
\[
m_0(p)\ge m_0(1/3)=\frac{106}{8775}>0
\quad(p<1/3),
\]
while $J_H\le1/100$ gives
\[
h(p)\ge\frac{353}{1500}>0.
\]
This is Condition~\ref{cond:c4}.

We now construct the bounds required by Condition~\ref{cond:c5}. Starting with the interval $[0,1]^2$, the candidate policy operator returns a pair of functions in the same interval. Hence $0\le W_k^\pi\le1$. Its equations imply $W_1=W_0+b$ below $1/3$, $W_1=\delta W_0$ in the band, and $W_1=W_0$ above $2/5$. The $x=0$ recursion first yields $\sup_pW_0(p)\le99/100$. Separating the null weak-trace self-loop then gives, on the $H$ face,
\[
L:=\frac{2065792}{2693925}\le W_0(p)\le
\frac{24693}{30700}=:U.
\]
For low trained beliefs define
\[
B^*=\frac{10499977}{18420000},
\qquad
\psi(p)=1-\frac p2-\frac72p^2.
\]
Let $c=106/8775$. An explicit suppression subsolution is
\[
\ell^D_0(p)=
\begin{cases}c,&p<1/3,\\L,&p\ge1/3,
\end{cases}
\qquad
\ell^D_1(p)=
\begin{cases}
0,&p<1/3,\\
\delta L,&1/3\le p<2/5,\\
L,&p\ge2/5.
\end{cases}
\]
An explicit suppression supersolution is
\[
u^D_0(p)=
\begin{cases}99/100,&p<1/3,\\U,&p\ge1/3,
\end{cases}
\]
\[
u^D_1(p)=
\begin{cases}
\psi(p),&p<2/7,\\
B^*,&2/7\le p<1/3,\\
\delta U,&1/3\le p<2/5,\\
U,&p\ge2/5.
\end{cases}
\]
Appendix~\ref{app:envelopes} verifies $\ell^D\le\cT^\pi\ell^D$, $\cT^\pi u^D\le u^D$, and in particular the primitive bounds
\[
W_1(p)\le
\begin{cases}
\psi(p),&p<2/7,\\
B^*,&2/7\le p<1/3.
\end{cases}
\]

Split the knowledge-sensitive region at $8/23$, where $p^+=2/5$, and $16/41$, where $p^-=2/7$. The resulting upper bounds on a deviation inducing standard implementation by an experienced worker are decreasing on their stated intervals and have maxima
\[
-\frac{97152803}{7183800000},\qquad
-\frac{4699067}{28501726500},\qquad
-\frac{954315461119}{18113951700000}.
\]
All are negative, proving the uniform suppression inequality.

For the focal preservation inequality, let $s=p^-(p_1)=125/477$ and define the feasible-plan floor
\[
F_0(p)=0,\qquad
F_{n+1}(p)=
\begin{cases}
m_1(p)+\delta\cP_MF_n(p),&p<1/3,\\
0,&p\ge1/3.
\end{cases}
\]
Exact rational recursion gives
\[
F_3(s)=
\frac{202938201927754323080507725397784860104567}
{1013557574011466310253154443299233282400000}
>\frac15.
\]
Using $L,U,\psi$, and $F_3$ gives the primitive bounds
\[
A_D(p_1)<
-\frac{143657902252609}{7794652163220000}<0,
\]
\[
A_K(p_1)>
\frac{9155142135908475778225424252164824152761486417}
{2691650665406785058355420797127836447674080000000}>0.
\]
For the focal standard implementation, take $f_1(125/301)=L$, $f_1(125/477)=F_3(125/477)$, and $u^K_1(p_1)=\delta U$. The first lower bound comes from $\ell^D$, the second from the displayed finite feasible plan, and the upper bound comes from $u^D$. Thus both parts of Condition~\ref{cond:c5} hold. The displayed bounds use only primitive recursions and rational arithmetic; they do not approximate a value function on a belief grid.
\end{proof}

\subsection{Off-path choices and local robustness}

On-path training is strict. Appendix~\ref{app:trembles} gives a general, rather than witness-specific, training fixed point after every worker-action tremble. Worker trembles to $\Hact$ and adviser trembles to $\Sact$ create no training opportunity. Completely mixed behavioral approximations then generate Bayes-consistent beliefs; they need not themselves be equilibria. The factorization lemma applies to every member of the approximating sequence. After the candidate policy is verified, the focal feasible-plan part of Condition~\ref{cond:c5} proves $A_K(p_1)>0$, while its suppression part proves $A_D(p_1)<0$. This completes equilibrium verification.

A small parameter change moves the worker's thresholds. Because implementation is discrete, we cannot infer global continuity of the policy from contraction alone. The following argument instead uses a finite tree from the initial belief and controls the remaining discounted tail. It preserves the focal behavior without requiring the same recommendation at every off-path belief.

\begin{lemma}
\label{lem:robust}
There is an open neighborhood $\mathcal N$ of \eqref{eq:W} such that, for every primitive vector in $\mathcal N$, some factorized Markov stationary sequential equilibrium has the following focal sequence from the fixed initial belief $p_0=5/16$ (and from every sufficiently nearby initial belief): warning and a valuable standard implementation, strict successor training, suppression after the favorable posterior, loss of competence, and resumed warning with intensive implementation. The one-step destructive comparison remains pivotal. Recommendations at other histories may change in the neighborhood.
\end{lemma}

\begin{proof}
At \eqref{eq:W}, posterior odds reachable from $p_0$ under standard implementation and weak-audit observations have the form
\[
\frac5{11}
\left(\frac54\right)^{n_c}
\left(\frac58\right)^{n_w}
\left(\frac32\right)^{n_1}
\left(\frac23\right)^{n_0},
\qquad n_c,n_w,n_1,n_0\in\mathbb N_0.
\]
After reduction, the prime $11$ remains in the denominator. Such odds can equal neither $1/2$ nor $2/3$, the odds at $p_R$ and $p_T$. Therefore, on every finite posterior tree, all worker choices are a positive distance from their cutoffs.

Fix a uniform nearby bound $\bar g$ on the absolute one-period adviser return. Truncating after $N$ opportunities changes any continuation value by at most
\[
\frac{\bar g\,\delta^N}{1-\delta}.
\]
Choose $N$ so that twice this bound is smaller than the minimum of the strict dynamic action gaps certified above: actual destruction, one-step preservation, initial warning, and resumed warning. On the resulting finite tree, posterior maps, probabilities, rewards, and the finite Bellman maximum are continuous in primitives; worker actions remain fixed for all sufficiently small perturbations because no node lies on a cutoff. Hence those dynamic signs persist, and the tail bound carries them to the infinite-horizon problem. The training gap, focal posterior crossing, and one-opportunity standard implementation-value gap are finite primitive inequalities and persist directly by continuity.

At remaining face states choose stationary optimal adviser actions. At interior training states, the adviser's mixed-action best-response correspondence and the atomless-cost training response are nonempty, convex-valued, and upper hemicontinuous; the construction in Appendix~\ref{app:trembles} is the explicit version under the certified face policy, and the same one-dimensional fixed-point argument supplies a completion after nearby perturbations. Standard measurable selection for the finite-action discounted problem then gives stationary strategies. Conditions~\ref{cond:c1} and \ref{cond:c2} remain strict nearby, so the focal crossing and on-path all-training conclusion also persist. This proves the stated open-neighborhood result without claiming global sup-norm continuity of the original face policy.
\end{proof}

For completeness, almost-sure termination of the intensive-implementation phase holds under the general conditions, not just at the numerical example. While this phase continues, the state has $x=0$ and $p\ge p_R$, and every released project induces $\Hact$. Conditional on $\theta=B$, each non-null observation changes log posterior odds by a fresh increment with mean
\[
a_B\log\frac{a_G}{a_B}
+(1-a_B)\log\frac{1-a_G}{1-a_B}<0.
\]
The inequality follows from $a_B\ne a_G$. Since $\varepsilon>0$, infinitely many such observations would arrive if the phase never ended. The strong law would then send log posterior odds to minus infinity, contradicting $p\ge p_R>0$. The phase therefore ends almost surely. The exact mean duration in example~\eqref{eq:W} is computed in Corollary~\ref{cor:meta}.

This completes the proof of Theorem~\ref{thm:main}.

\section{Exact envelope algebra}
\label{app:envelopes}

The $x=0$ fixed-point equation and $0\le W_k\le1$ imply
\[
\sup_pW_0(p)
\le\max\left\{\frac{29}{100}+\delta\left(\sup_pW_0(p)+\frac1{100}\right),
\frac6{25}+\delta\sup_pW_0(p)\right\},
\]
and hence $\sup_pW_0(p)\le99/100$. Since $J_H\le1/100$, separating the $\delta(1-\varepsilon)W_0(p)$ self-loop gives
\[
L=
\frac{353/1500+\delta\varepsilon(106/8775)}{1-\delta(1-\varepsilon)},
\qquad
U=
\frac{6/25+\delta\varepsilon(99/100)}{1-\delta(1-\varepsilon)}.
\]

The constants above can now be used to check the comparison inequalities. On a low state, $m_0\ge c$ and $m_1=m_0+b>0$, so the low components of $\cT^\pi\ell^D$ dominate $\ell^D$. On an $H$ state,
\[
h(p)+\delta\cP_H\ell^D_0(p)
\ge\frac{353}{1500}+\delta\{(1-\varepsilon)L+\varepsilon c\}=L.
\]
The band component satisfies $(\cT^\pi\ell^D)_1=\delta L$. Hence $\ell^D\le\cT^\pi\ell^D$.

For the upper vector, a low untrained state satisfies
\[
m_0(p)+\delta\cP_Mu^D_1(p)\le\frac{29}{100}+\frac7{10}=\frac{99}{100}.
\]
On an $H$ state, $h\le6/25$, the null successor remains on the $U$ region, and every non-null successor has $u^D_0\le99/100$. Therefore
\[
h(p)+\delta\cP_Hu^D_0(p)
\le\frac6{25}+\delta\{(1-\varepsilon)U+\varepsilon(99/100)\}=U.
\]
The band and high trained components then follow from the policy equations. The two remaining checks concern an experienced worker at low reputations.

We first specify which interval contains each boundary point. The worker tie break assigns $p_R$ to $\Hact$ on the untrained face and $p_T$ to $\Hact$ on the trained face. The relevant posterior preimages satisfy
\[
p^+(8/33)=2/7,\qquad p^+(2/7)=1/3,
\]
\[
p^+(8/23)=2/5,\qquad p^-(16/41)=2/7.
\]
These assignments keep the comparison functions consistent with the worker's tie-breaking rule.

For $2/7\le p<1/3$, the favorable successor is in the knowledge-sensitive region and $u^D_1\le1$ on the unfavorable successor. Directly evaluating the supersolution operator gives
\[
(\cT^\pi u^D)_1(p)
\le m_1(p)+\delta\{q(p)\delta U+[1-q(p)]\}\le B^*=u^D_1(p),
\]
where the last expression is decreasing and equals $B^*$ at $2/7$. For $p\le8/33$, direct substitution gives the $\psi$-supersolution slack
\[
\psi(p)-m_1(p)-\delta\cP_M\psi(p)
=\frac{7p^2(86p^2-166p+65)}{4(4+p)(8-3p)}\ge0,
\]
with equality only at $p=0$. At $p=8/33$, the actual favorable-branch bound is $B^*<\psi(2/7)$, so the displayed calculation remains conservative. For $8/33\le p<2/7$, use $B^*$ on the favorable branch and $\psi$ on the unfavorable branch; clearing positive denominators yields slack at least
\[
\frac{21083905987}{936104400000}>0.
\]

The upper bounds on the three possible continuation configurations are
\[
E_1(p)=m_1(p)+\delta\{q(p)\delta U+[1-q(p)]\}-\delta L,
\]
\[
E_2(p)=m_1(p)+\delta\{q(p)U+[1-q(p)]\psi(p^-(p))\}-\delta L,
\]
\[
E_3(p)=m_1(p)+\delta\{q(p)U+[1-q(p)]B^*\}-\delta L.
\]
They apply on $[1/3,8/23)$, $[8/23,16/41)$, and $[16/41,2/5)$, respectively. In particular, $8/23$ belongs to the second region because its favorable successor is exactly $p_T$, and $16/41$ belongs to the third because its unfavorable successor is exactly $2/7$. On the whole band, $J_M'(p)>0$ because
\[
J_M'(p)=
\frac{-6p(p-1)(3p^3+6p^2-66p+32)}{(p+4)^2(3p-8)^2}>0,
\]
where the cubic is bounded below by $844/125$. Hence $m_1'(p)<-3/20$. The first envelope has an additional negative derivative. For the second, $p^-(p)\in[1/4,2/7]$, $\psi' <0$, and
\[
\delta q'(p)[U-\psi(p^-(p))]
\le\frac{150153}{6140000}<\frac1{40}.
\]
For the third,
\[
\delta q'(p)(U-B^*)
=\frac{30210761}{1228000000}<\frac1{40}.
\]
Each $E_j$ is therefore strictly decreasing on its interval. Its left endpoint supplies the largest possible deviation gain, giving the three negative rational bounds in Appendix~\ref{app:verification}.

\section{Learning after off-path implementation}
\label{app:trembles}

Every standard implementation at $p<p_T$, including a worker tremble by an untrained worker in the knowledge-sensitive region, satisfies the strict-renewal argument in Lemma~\ref{lem:renewal}. It therefore sends $x'$ to one. It remains to complete a worker tremble to $\M$ at $p\ge p_T$. Write $s=p^-(p)$. The confirming posterior is above $p_T$, where knowledge does not change implementation. Any private benefit from learning must therefore come from the disconfirming branch.

Let
\[
F_\kappa(v)=\min\left\{\frac{v}{\bar\kappa},1\right\}
\]
be the training share generated by a gross private return $v$.

If $s<p_R$, the adviser warns for every competence probability at the unfavorable successor and competence saves rebuild cost $b$. The unique training share is therefore
\[
x'=F_\kappa([1-q(p)]b).
\]
This probability need not equal one. Sequential rationality requires the optimal training share at this history, not full renewal at every possible tremble.

If $s\in[p_R,p_T)$, define
\[
C(p)=[1-q(p)][U_M(s,1)-U_H]>0,
\qquad
\varphi(p)=F_\kappa(C(p)).
\]
Set $x'=\min\{\varphi(p),\widehat x(s)\}$. If $\varphi(p)\le\widehat x(s)$, let the adviser warn at the unfavorable successor; then $x'=F_\kappa(C(p))$ is the worker's training response. If $\varphi(p)>\widehat x(s)$, put the public state at the adviser's cutoff and let the indifferent adviser warn with probability
\[
\lambda(p)=\frac{\bar\kappa\widehat x(s)}{C(p)}\in(0,1).
\]
Then
\[
F_\kappa(\lambda(p)C(p))=\widehat x(s)=x'.
\]
Finally, if $s\ge p_T$, competence changes neither implementation choice and the unique training share is $x'=0$.

At the witness, $p^-(4/9)=p_R$ and $p^-(16/31)=p_T$. Moreover, for $p_T\le p<4/9$, $[1-q(p)]b>\bar\kappa$, so the first case happens to give full training. In the middle case,
\[
C(p)=\frac{3(16-31p)}{2000},
\qquad
\lambda(p)=\frac{\widehat x(s)}{400C(p)},
\]
giving explicit values for the interior completion in the example.

A worker tremble to $\Hact$, or a history following $\Sact$, offers no training and sends $x'$ to zero. Adviser trembles to $\R$ are followed by the worker cutoff rule and are covered by the same cases. To obtain consistency, perturb every adviser and implementation action with positive probability and approximate each training cutoff by a completely mixed behavioral rule for every cost type. Bayes' rule then determines beliefs at every history. These perturbations converge to the assessment just constructed; sequential equilibrium requires consistency of this approximating sequence, not optimality of the perturbed profiles themselves.

\section{Exact preservation-economy certificate}
\label{app:preservation}

For face $k\in\{0,1\}$, let $\tau_k=p_R$ if $k=0$ and $\tau_k=p_T$ if $k=1$. The transition law in this counterfactual is imposed mechanically: $\Sact$ preserves $k$, $\M$ sends the successor to face one, and $\Hact$ sends the successor to face zero. Accordingly,
\[
(\widetilde{\cB}V)_k(p)=\max\{S_k(p;V),R_k(p;V)\},
\]
where
\[
S_k(p;V)=\beta r(p)+\delta V_k(p),
\]
and, if $p<\tau_k$,
\[
R_k(p;V)=U_M(p,k)+\cP_M[\beta r+\delta V_1](p),
\]
while, if $p\ge\tau_k$,
\[
R_k(p;V)=U_H+\cP_H[\beta r+\delta V_0](p).
\]
Every expected current return under an available induced action is nonnegative. It is also at most $143/10$, so
\[
\widetilde{\cB}(143/3)\le143/3.
\]
The accompanying Python certificate evaluates this operator on the finite posterior tree. Posterior updates, maximization, and recursion use rational numbers throughout; cached values avoid repeating identical continuation calculations. It checks
\[
R_1(p;\widetilde{\cB}^{18}0)
-S_1(p;\widetilde{\cB}^{18}(143/3))
\]
at $p_0$ and $p_1$ against the exact rational thresholds $1/20$ and $2/25$. The resulting lower bounds on the two release margins are approximately $0.055781070750$ and $0.081834618595$. Decimal values are reported only for readability. The inequalities tested by the program are exact rational comparisons on a finite tree.


\begin{thebibliography}{99}
\small

\bibitem[Achim and Knoepfle(2025)]{AchimKnoepfle2025}
Achim, Peter, and Jan Knoepfle. 2025. ``The Tension between Trust and Oversight in Long-term Relationships.'' Working paper, \href{https://arxiv.org/abs/2504.02696}{arXiv:2504.02696}.

\bibitem[Andina-D\'iaz and Garc\'ia-Mart\'inez(2020)]{AndinaGarcia2020}
Andina-D\'iaz, Ascensi\'on, and Jos\'e A. Garc\'ia-Mart\'inez. 2020. ``Reputation and News Suppression in the Media Industry.'' \emph{Games and Economic Behavior} 123: 240--271.

\bibitem[Bar-Isaac and Deb(2021)]{BarIsaacDeb2021}
Bar-Isaac, Heski, and Joyee Deb. 2021. ``Reputation with Opportunities for Coasting.'' \emph{Journal of the European Economic Association} 19(1): 200--236.

\bibitem[Camara and Dupuis(2023)]{CamaraDupuis2023}
Camara, Fanny, and Nicolas Dupuis. 2023. ``Avoiding Judgement by Recommending Inaction: Beliefs Manipulation and Reputational Concerns.'' Working paper, SSRN 4663047.

\bibitem[Carlin and Manso(2011)]{CarlinManso2011}
Carlin, Bruce I., and Gustavo Manso. 2011. ``Obfuscation, Learning, and the Evolution of Investor Sophistication.'' \emph{Review of Financial Studies} 24(3): 754--785.

\bibitem[Che and H\"orner(2018)]{CheHorner2018}
Che, Yeon-Koo, and Johannes H\"orner. 2018. ``Recommender Systems as Mechanisms for Social Learning.'' \emph{Quarterly Journal of Economics} 133(2): 871--925.

\bibitem[Chen et~al.(2026)]{ChenEtAl2026}
Chen, Yi, Kai Du, Phillip Stocken, and Zhe Wang. 2026. ``Peer Learning, Enforcement, and Reputation.'' \emph{RAND Journal of Economics} 57(3): 648--662.

\bibitem[Deb and Ishii(2025)]{DebIshii2025}
Deb, Joyee, and Yuhta Ishii. 2025. ``Reputation Building under Uncertain Monitoring.'' \emph{Theoretical Economics} 20(1): 169--208.

\bibitem[Garfagnini and Strulovici(2016)]{GarfagniniStrulovici2016}
Garfagnini, Umberto, and Bruno Strulovici. 2016. ``Social Experimentation with Interdependent and Expanding Technologies.'' \emph{Review of Economic Studies} 83(4): 1579--1613.

\bibitem[Halac and Kremer(2020)]{HalacKremer2020}
Halac, Marina, and Ilan Kremer. 2020. ``Experimenting with Career Concerns.'' \emph{American Economic Journal: Microeconomics} 12(1): 260--288.

\bibitem[Hauser(2023)]{Hauser2023}
Hauser, Daniel N. 2023. ``Censorship and Reputation.'' \emph{American Economic Journal: Microeconomics} 15(1): 497--528.

\bibitem[Kuang et~al.(2026)]{KuangEtAl2026}
Kuang, Zhonghong, Yi Liu, and Dong Wei. 2026. ``Incentivizing Knowledge Transfers.'' Working paper, \href{https://arxiv.org/abs/2507.11018}{arXiv:2507.11018}, revised January 2026.

\bibitem[Liu(2011)]{Liu2011}
Liu, Qingmin. 2011. ``Information Acquisition and Reputation Dynamics.'' \emph{Review of Economic Studies} 78(4): 1400--1425.

\bibitem[Lukyanov and Vlasova(2026)]{LukyanovVlasova2026}
Lukyanov, Georgy, and Anna Vlasova. 2026. ``Endogenous Vindication: Reputation and Effort in Expert Advice.'' Working paper, \href{https://arxiv.org/abs/2508.19676}{arXiv:2508.19676}, revised July 2026.

\bibitem[Lukyanov(2026)]{LukyanovEffort2026}
Lukyanov, Georgy. 2026. ``Effort without Evidence.'' Working paper.

\bibitem[Marinovic and Szydlowski(2022)]{MarinovicSzydlowski2022}
Marinovic, Iv\'an, and Martin Szydlowski. 2022. ``Monitoring with Career Concerns.'' \emph{RAND Journal of Economics} 53(2): 404--428.

\bibitem[Marinovic and Szydlowski(2023)]{MarinovicSzydlowski2023}
Marinovic, Iv\'an, and Martin Szydlowski. 2023. ``Monitor Reputation and Transparency.'' \emph{American Economic Journal: Microeconomics} 15(4): 1--67.

\bibitem[Min(2025)]{Min2025}
Min, Weicheng. 2025. ``Bad Reputation Due to Incompetent Expert.'' \emph{Journal of Economic Theory} 230: 106080.

\bibitem[Pei(forthcoming)]{Pei2026}
Pei, Harry. Forthcoming. ``Reputation Effects with Endogenous Records.'' \emph{American Economic Journal: Microeconomics}. \href{https://doi.org/10.1257/mic.20250318}{doi:10.1257/mic.20250318}.

\bibitem[Prendergast and Stole(1996)]{PrendergastStole1996}
Prendergast, Canice, and Lars Stole. 1996. ``Impetuous Youngsters and Jaded Old-Timers: Acquiring a Reputation for Learning.'' \emph{Journal of Political Economy} 104(6): 1105--1134.

\bibitem[Strulovici(2022)]{Strulovici2022}
Strulovici, Bruno. 2022. ``Can Society Function Without Ethical Agents? An Informational Perspective.'' Working paper, arXiv:2003.05441.


\end{thebibliography}
\end{document}